\documentclass[a4paper,UKenglish,cleveref]{lipics-v2021}
\usepackage{blindtext}
\usepackage{hyperref}

\usepackage{mathtools}
\usepackage[dvipsnames]{xcolor}
\usepackage{tikz}
\usetikzlibrary{decorations.pathmorphing, arrows.meta}
\usetikzlibrary{hobby}

\theoremstyle{definition}

\crefname{claim}{Claim}{Claims}

\newcommand{\NN}{\ensuremath \mathbb{N}}

\newcommand{\preTreeDec}{pre-tree decomposition}
\newcommand{\PreTreeDec}{Pre-tree decomposition}
\newcommand{\submod}{submodular}

\DeclareMathOperator{\wid}{wd}
\DeclareMathOperator{\dep}{dp}
\DeclareMathOperator{\partitions}{Part}
\DeclareMathOperator*{\CR}{CR}
\DeclareMathOperator*{\monCR}{mon-CR}
\newcommand{\CRparam}[3][G]{\ensuremath{\CR_{#3}^{#2}(#1)}}
\newcommand{\monCRparam}[3][G]{\ensuremath{\monCR_{#3}^{#2}(#1)}}
\newcommand{\CRkq}[1][G]{\CRparam[#1]{k}{q}}
\newcommand{\monCRkq}[1][G]{\monCRparam[#1]{k}{q}}
\newcommand{\complementOf}[1]{\ensuremath{\overline{#1}}}
\newcommand{\complementOfB}[1]{\complementOf{#1}}
\newcommand{\ext}{\ensuremath{\rightarrow}}
\newcommand{\considered}{considered}
\newcommand{\consider}{consider}
\newcommand{\I}{\ensuremath{I}}

\makeatletter
\def\namedlabel#1#2{\begingroup
	#2%
	\def\@currentlabel{#2}%
	\phantomsection\label{#1}\endgroup
}
\makeatother

\title{Monotonicity of the cops and robber game for bounded depth
treewidth}

\author{Isolde Adler}{University of Bamberg, Germany}{isolde.adler@uni-bamberg.de}{https://orcid.org/0000-0002-9667-9841}{}

\author{Eva Fluck}{RWTH Aachen University, Germany}{fluck@cs.rwth-aachen.de}{https://orcid.org/0000-0002-9643-6081}{}

\authorrunning{I. Adler and E. Fluck}

\Copyright{Isolde Adler, and Eva Fluck}

\ccsdesc[500]{Mathematics of computing~Graph theory}

\keywords{tree decompositions, treewidth, treedepth, cops-and-robber game, monotonicity, homomorphism distinguishing closure}

\hideLIPIcs
\nolinenumbers

\begin{document}

\maketitle

\begin{abstract}
	We study a variation of the cops and robber game characterising treewidth,
where in each play at most $q$ cops can be placed in order to catch the robber, 
	where $q$ is a parameter of the game.
We prove that if $k$ cops have a winning strategy in this game, then $k$ cops have a
monotone winning strategy. As a corollary we obtain a new characterisation of
bounded depth treewidth, and we give a positive answer to an open question 
	by Fluck, Seppelt and Spitzer (2024), thus showing that 
graph classes of bounded depth treewidth
are homomorphism distinguishing closed.
	
Our proof of monotonicity
substantially reorganises a winning strategy by first transforming it into
a pre-decomposition, which is inspired by decompositions of matroids, and then applying an intricate 
	breadth-first `cleaning up' 
procedure along the pre-decomposition (which may temporarily lose the property of representing a strategy), in order to achieve
monotonicity while controlling the number of cop placements simultaneously 
	across all branches of the decomposition via a vertex exchange argument. We believe this can be useful in future research.
\end{abstract}

\section{Introduction}

Search games were introduced by Parsons and Petrov 
in \cite{Parsons78,parsons78search,petrov82} and since then gained 
much interest in many (applied and theoretical) areas of computer science and in discrete 
mathematics \cite{AignerF84,BodlaenderT04,BienstockS91,LaPaugh93,FranklinGY00,Obdrzalek06,HunterK08,FominGK08,HollingerKS10,GroheM14}.
In search games on graphs, a fugitive and a 
set of searchers move on a graph, according to 
given rules. The searchers' goal is to capture the
fugitive, and the fugitive tries to escape. Here the interest lies in minimising the resources needed to guarantee capture. Typically this means minimising the number of searchers, but we also seek to bound the number of new placements of searchers.
Search games have proven very useful for providing a deep understanding of structural and 
algorithmic properties of width parameters of graphs,
such as treewidth~\cite{Bienstock89,SeymourT93}, 
pathwidth~\cite{BienstockS91}, cutwidth~\cite{MakedonS89}, and 
directed treewidth~\cite{JohnsonRST01}, treedepth~\cite{NesetrilM06}, and $b$-branched treewidth~\cite{FominFN09,MazoitN08}. 

The crux in relating a given variant of a search game to a width parameter often lies in the
question of whether the game is \emph{monotone}, i.\,e.~whether the searchers always have a winning strategy 
in which a previously cleared area never needs to be searched again -- without needing additional resources.
   Furthermore, monotonicity of a search game provides a polynomial space certificate for proving that
determining the winner is in NP.

In their classic paper~\cite{SeymourT93}, Seymour and Thomas proved monotonicity of the cops and robber game characterising treewidth. 
They use a very elegant inductive argument via the 
dual concept of \emph{brambles}.
In this paper we study a variation of this game, where $k$ cops try to capture a robber, but they are limited to making at most $q$ cop placements, for a fixed number $q\in \mathbb N$. 
It is an open question from~\cite{FluckSS24}, whether this game is monotone. We give a positive answer to this question.

The notion of \emph{treedepth} was first introduced by Ne\v{s}et\v{r}il and Ossona de Mendes~\cite{NesetrilM06}. They exhibit a number of equivalent parameters, and a characterisation by a monotone game is implicitely given. This was subsequently made more explicit in~\cite{GiannopoulouT11}. In~\cite{GiannopoulouHT12}, a characterisation by a different game called \emph{lifo game} is given for which monotonicity is proven. The game we study can be seen as generalising the monotone game implicit in~\cite{NesetrilM06}. However, it is strictly more general and it is not monotone by definition.

Recently, width parameters received a renewed interest in the context of counting homomorphisms and the expressive power of logics~\cite{dvorak_recognizing_2010,grohe_counting_2020,dawar_lovasz-type_2021,ScheidtS23,FluckSS24}. In this context a non-monotone search game characterisation of the width parameter is useful to ensure that there are no graphs of higher width that can be added to the graph class without changing the expressive power \cite{neuen_homomorphism-distinguishing_2023,FluckSS24}. The main obstacle then is to find such a non-monotone characterisation, as the natural characterisation as a search-game of many graph parameters is inherently monotone. \emph{Bounded depth treewidth} and the game studied in this paper were first defined in~\cite{FluckSS24}. An equivalent characterisation of these graph classes by so-called $k$-pebble forest covers of depth $q$, which is bounded width tree depth, was already given in~\cite{Abramsky_pebbling_2017}.

\subparagraph{Homomorphism counts.}
Homomorphism counts are an emerging tool to study equivalence relations between graphs.
Many equivalence relations between graphs can be characterized as homomorphism indistinguishability relations, these include graph isomorphism \cite{lovasz_operations_1967}, graph isomorphism relaxations \cite{mancinska_quantum_2020, grohe_homomorphism_2022, roberson_lasserre_2023}, cospectrality \cite{dell_lovasz_2018} and equivalence with respect to first-order logic with counting quantifiers \cite{dvorak_recognizing_2010, grohe_counting_2020, dawar_lovasz-type_2021, FluckSS24}.
In order to study the expressiveness of such equivalence relations, it is crucial to know under which circumstances distinct graph classes yield distinct equivalence relations.
Towards this question one considers the closure of a graph class under homomorphism indistinguishability.
Let $\mathcal{F}$ be a graph class.
Two graphs $G, H$ are \emph{homomorphism indistinguishable over $\mathcal{F}$}, if for all $F\in\mathcal{F}$ the number of homomorphisms from $F$ to $H$ equals the number of homomorphisms from $F$ to $G$.
The graph class $\mathcal{F}$ is \emph{homomorphism distinguishing closed}, if for every graph $F\notin\mathcal{F}$ there exists two graphs $G, H$, that are homomorphism indistinguishable over $\mathcal{F}$ but that do not have the same number of homomorphisms from $F$.
It has been conjectured by Roberson \cite{Roberson_oddomorphisms_2022}, that all graph classes that are closed under taking minors and disjoint unions are homomorphism distinguishing closed.
So far the list of graph classes for which the conjecture is confirmed is short: the class of all planar graphs \cite{Roberson_oddomorphisms_2022}, graph classes that are essentially finite \cite{seppelt_logical_2023}, the classes of all graphs of tree width at most $k-1$ \cite{neuen_homomorphism-distinguishing_2023} and the classes of all graphs of tree depth at most $q$ \cite{FluckSS24}.
The latter two results rely on characterisations of the graph classes in terms of non-monotone cops-and-robber games.
We study \emph{bounded depth treewidth}, which bounds both the width and the depth simultaneously.
We give a game characterisation that does not rely on monotonicity, and as a consequence we obtain that graph classes of bounded depth treewidth are also homomorphism distinguishing closed.

\subparagraph{Our contribution.}

We show the following (cf.~Theorem~\ref{thm:equivalence}).\\[-.4cm]

\noindent \emph{Fix integers $k,q\geq 1$. For every graph $G$ the following are equivalent.
\begin{itemize}
 \item $G$ has a tree decomposition of width at most $k-1$ and depth at most $q$.
 \item $k$ cops have a monotone winning strategy in the cops and robber game on $G$ with at most $q$ placements.
 \item $k$ cops have a winning strategy in the cops and robber game on $G$ with at most $q$ placements.
\end{itemize}
}
\noindent The equivalence between the last two statements gives a positive answer to an open question from~\cite{FluckSS24}.
Our proof of monotonicity gives both a proof of monotonicity for the classical cops and robber game characterising treewidth as well as for the game characterising treedepth as special cases.
As a corollary, we obtain the following (cf.~Theorem~\ref{thm:exactness}).\\[-.3cm]

\noindent\emph{Let $k,q\geq 0$ be integers.
The class of graphs having a tree decomposition of width at most $k-1$ and depth at most $q$ is homomorphism distinguishing closed.
}

\subparagraph{Proof techniques.}
In contrast to the proof of monotonicity of the classic cops and robber game~\cite{SeymourT93}, our proof does not use a dual concept such as brambles. Instead, we modify a (possibly non-monotone) winning strategy, turning it first into what we call a \emph{pre-decomposition}, and then cleaning it up while keeping track of width and depth, thus finally transforming the pre-decomposition
into a monotone winning strategy.
Our concept of pre-decomposition is inspired by decompositions of matroids and it is based on ideas from~\cite{abs-0906-3857,AminiMNT09}. 
Our cleaning-up technique is similar to the proof of monotonicity of the the game for $b$-branching treewidth~\cite{MazoitN08}. However,
the cleaning-up technique in~\cite{MazoitN08} loses track of the number of cop placements, as local modifications may have 
non-local effects that are not controlled. We need to keep track in order to control the depth. 

This poses a major challenge which we resolve in our proof by a fine grain cleaning-up technique in our pre-decompositions based on a careful decision of which vertices to `push up and through the tree' and which to `push down'. 
The vertices `pushed up' may have an effect on the part of the pre-decomposition that was processed in previous steps, which we manage to control by a vertex exchange argument.
Additionally we keep track of how the first modification at some node in the pre-decomposition relates back to the original strategy.
We believe that our techniques will also help in future research.

Our proof provides an independent proof of monotonicity of the classic game characterising tree-width as a special case, namely when $q$ is greater than or equal to the number of vertices of the graph. Our proof strategy is entirely different, as it does not use an equivalence via a dual object such as brambles. Instead, we provide a more direct transformation of a (possibly non-monotone) winning strategy.

\subparagraph{Further related research.}
 Search games are used to model a variety of
real-world problems such as searching a lost person in a
system of caves~\cite{Parsons78}, clearing contaminated tunnels~\cite{LaPaugh93}, 
searching environments in robotics~\cite{HollingerKS10}, and modelling bugs in
distributed environments~\cite{FranklinGY00}, cf.~\cite{FominT08} for a survey.

 There is a fine line between games that are monotone and those that are not. For example, the marshalls and robber game played on a hypergraph is a natural generalisation of the cops and robber game, it is related to hypertree-width, but it is not 
 monotone~\cite{Adler04}. However, the monotone and the non-monotone variants are strongly related~\cite{AdlerGG07} to eachother.  

 \subparagraph{Structure of the paper.}
In Section~\ref{sec:preliminaries} we fix our notation and we define tree decompositions of bounded depth and width.  
Section~\ref{sec:predec} introduces pre-tree decompositions, relevant properties, and establishes a relation to tree decompositions. The game is introduced in Section~\ref{sec:game}, and in Section~\ref{sec:make-exact} we give the main proofs, showing how to make a strategy tree exact while maintaining the bounds on width and depth.
The insights given by our answer to the open question in the area of homomorphism counts are briefly discussed in~\cref{sec:hom}.

\section{Preliminaries}\label{sec:preliminaries}

\subparagraph*{Sets and partitions}
Let $A$ be a finite set.
We write $2^A$ to denote the power-set of $A$ and, for $k\in \NN$, $\binom{A}{\leq k}$ to denote all subsets of $A$ of size $\leq k$.
$\partitions(A)$ is the set of all \emph{partitions} of $A$, where we allow partitions to contain multiple (but finite) copies of the empty set.
Let $P=\{X_1,\ldots,X_d\}\in\partitions(A)$ and $F\subseteq A$.
For $i\in[d]$, the partition
\[
P_{X_i\ext F}\coloneqq \{X_1\setminus F, \ldots, X_{i-1}\setminus F, X_i\cup F,X_{i+1}\setminus F,\ldots,X_d\setminus F^i \},
\]
is called the \emph{$F$-extension in $X_i$ of $P$}.
A function $w\colon \partitions(A)\rightarrow \NN$ is \emph{\submod} if, for all $P,Q\in \partitions(A)$, for all sets $X\in P$ and $Y\in Q$ with $X\cup Y \neq A$, it holds that
\[w(P)+w(Q)\geq w(P_{X\ext \complementOf{Y}})+w(Q_{Y\ext \complementOf{X}}).\]
Let $f\colon A\rightarrow B$ be a function and $C\subseteq A$.
By $f|_C$ we denote the restriction of $f$ to $C$, i.\,e.~$f|_C\colon C\rightarrow B$ and $f|_C(c) = f(c)$, for all $c\in C$.

\subparagraph*{Graphs}
A graph $G$ is a tuple $(V(G), E(G))$, where $V(G)$ is a finite set of vertices and $E(G) \subseteq \binom{V(G)}{\leq 2}$ is the set of edges.
We usually write $uv$ or $vu$ to denote the edge $\{u, v\} \in E(G)$.
If \(G\) is clear from the context we write \(V,E\) instead of \(V(G),E(G)\).
We write $\I(G)$ to denote the set of isolated vertices in $G$, that is for every $v\in\I(G)$, there is no $u\in V(G)$ with $u\neq v$ and $uv\in E(G)$.
By $G^\circ$ we denote the graph obtained from $G$ by adding all self-loops that are not present in $G$, that is $V(G^\circ)\coloneqq V(G)$ and $E(G^\circ)\coloneqq E(G) \cup \{vv\mid v\in V(G)\}$.
For $v\in V$ we write $E_G(v)\coloneqq\{uv\mid uv\in E(G)\}$ for the edges incident to $v$.

A \emph{tree} is a graph where any two vertices are connected by exactly one path.
A \emph{rooted tree} $(T, r)$ is a tree $T$ together with some designated vertex $r \in V(T)$, the \emph{root} of $T$.
By \(L(T)\) we denote the set of all \emph{leaves} of \(T\), that is \(L(T)\coloneqq\{ v\in V(T) \mid |N(v)|=1\}\).
All vertices that are not leafs are called \emph{inner vertices}.

At times, the following alternative definition is more convenient.
We can view a rooted tree $(T, r)$ as a pair $(V(T), \preceq)$, where $\preceq$ is a partial order on $V(T)$ and for every $v \in V(T)$ the elements of the set $\{u \in V(T) \mid u \preceq v\}$ are pairwise comparable: The minimal element of $\preceq$ is precisely the root of $T$, and we let $v \preceq w$ if $v$ is on the unique path from $r$ to $w$.
Let $t,t'\in V(T)$, we call $t^*\in V(T)$ the \emph{greatest common ancestor} if $t^*\preceq t,t'$ but for all $t''\in V(T)$ with $t^*\prec t''$ either $t''\not\preceq t$ or $t''\not\preceq t'$.

\begin{definition}
	\label{def:tree_decomp}
	Let $G$ be a graph, let $(T,r)$ be a rooted tree and let $\beta\colon V(T)\rightarrow 2^{V(G)}$ be a function from the nodes of $T$ to sets of vertices of $G$.
	We call $(T,r,\beta)$ a \emph{tree decomposition} of $G$, if
	\begin{description}
		\item[\namedlabel{ax:TreeDec1}{(T1)}] $\bigcup_{t\in V(T)} G[\beta(t)] = G$, and
		\item[\namedlabel{ax:TreeDec2}{(T2)}] for every vertex $v \in G$, the graph $T_v\coloneqq T[\{t\in V(T)\mid v\in \beta(t)\}]$ is connected.
	\end{description}
	The sets $\beta(t)$ are called the \emph{bags} of this tree decomposition.
\end{definition}

The \emph{width} of a tree decomposition $(T,r,\beta)$ is $\wid(T,r,\beta)\coloneqq\max_{t\in V(T)}|\beta(t)|-1$, the \emph{depth} is $\dep(T,r,\beta)\coloneqq\max_{\ell\in L(T)} |\bigcup_{t\in P_\ell} \beta(t)|$, where $P_\ell$ is the path from $\ell$ to the root.
The \emph{tree width} of a graph $G$ is the minimum width of any tree decomposition of $G$, the \emph{tree depth} of a graph $G$ is the minimum depth of any tree decomposition (see cf~\cite{FluckSS24}).
For $k,q\geq 1$ we define the class $\mathcal{T}^k_q$ to be all graphs that have a tree decomposition $(T,r,\beta)$ with $\wid(T,r,\beta)\leq k-1$ and $\dep(T,r,\beta)\leq q$.
The following lemma is a well known consequence from \ref{ax:TreeDec2}.

\begin{lemma}
	\label{lem:tdec-conn}
	Let $G$ be a graph and $U\subseteq V(G)$ connected in $G$.
	Let $(T,r,\beta)$ be a tree decomposition of $G$, then
	$T_U\coloneqq T[\{t\in V(T)\mid U\cap \beta(t)\neq \emptyset\}]$ is connected.
\end{lemma}

\section{\PreTreeDec, exactness and submodularity}\label{sec:predec}

Here we consider a definition of tree decompositions that is inspired by matroid tree decompositions.
We relax this definition into what we call a \preTreeDec.

\begin{definition}
	Let $G=(V(G),E(G))$ be a graph.
	Let $X\subseteq E(G)$.
	We define $\delta(X)\coloneqq\{v\in V(G)\mid \exists e\in X, e'\in E(G)\setminus X, v\in e\cap e' \}$.
	Let $\pi$ be a partition of $E(G)$.
	We define
	\[
	\delta(\pi)\coloneqq\{v\in V(G)\mid \exists X\in \pi, v\in\delta(X) \}.
	\]
	A tuple $(T,r,\beta,\gamma)$, where $(T,r)$ is a (rooted) tree, $\beta\colon V(T)\rightarrow 2^{V(G)}$ and $\gamma\colon \overrightarrow{E(T)}\rightarrow 2^{E(G)}$, is a \emph{(rooted) \preTreeDec} if:
	\begin{description}
		\item[\namedlabel{ax:preTreeRoot}{(PT1)}] $\beta(r)=\emptyset$ and for every connected component $C$ of $G$, there is a child $c$ of the root with $\gamma(r,c)=E(C)$.
		\item[\namedlabel{ax:preTreeLeaf}{(PT2)}] For every leaf $\ell\in L(T)$ with neighbour $t$, it holds that $|\gamma(t,\ell)|\leq 1$.
		\item[\namedlabel{ax:preTreePart}{(PT3)}] For every internal node $t\in V(T)\setminus L(T)$, we define $\pi_t\coloneqq(\gamma(t,t_1),\ldots,\gamma(t,t_d))$, where $N(t)=\{t_1,\ldots,t_d\}$ an arbitrary enumeration of the neighbours of $t$, and for a leaf $\ell \in L(T)$ with parent $p$ we define $\pi_\ell\coloneqq(\gamma(\ell,p),\complementOf{\gamma(\ell,p)})$.
		For every $t\in V(T)$, the tuple $\pi_t$ is a partition of $E(G)$ and $\beta(t)\supseteq \delta(\pi_t)$.
		\item[\namedlabel{ax:preTreeEdge}{(PT4)}] For every edge $st\in E(T)$, it holds that $\gamma(s,t)\cap \gamma(t,s) = \emptyset$.
	\end{description}
	We call an edge $st\in E(T)$ \emph{exact} if $\gamma(s,t)\cup\gamma(t,s)=E(G)$, we call $(T,r,\gamma,\beta)$ \emph{exact}, if every edge is exact and $\beta(t)= \delta(\pi_t)$, for all $t\in V(T)$.
	We call $\beta(t)$ the \emph{bag} at node $t$ and $\gamma(s,t)$ the \emph{cone} at edge $st$.
\end{definition}

\begin{observation}
	\label{obs:exact-subtree}
	Let $(T',r)$ be a subtree of $(T,r)$ with the same root. If all edges in $T'$ are exact then $\{\gamma(t,\ell)\mid \ell\in L(T'), t \text{ parent of } \ell\}$ is a partition of $E(G)$.
\end{observation}

Similar to the definition of width and depth for tree decompositions we define the width and depth of a \preTreeDec.

\begin{definition}
	The \emph{width} of a partition $\pi$ of the edges of a graph is
	\[
	\wid(\pi)\coloneqq|\delta(\pi)|.
	\]
	The \emph{width} of a \preTreeDec\ is
	\[
	\wid(T,r,\beta,\gamma)\coloneqq\max_{t\in V(T)} |\beta(t)| - 1.
	\]
	The \emph{depth} of a rooted \preTreeDec\ is
	\[
	\dep(T,r,\beta,\gamma) \coloneqq \max_{t\in V(T)} \sum_{s\in P_t\setminus\{r\}} |\beta(s) \setminus \beta(p_s)|,
	\]
	where $P_t$ is the unique path from the root $r$ to $t$ and $p_s$ is the parent of $s$.
\end{definition}

The reader may note that the width of a \PreTreeDec\ only gets smaller if one sets $\beta(t)\coloneqq \delta(\pi_t)$, for all nodes $t\in V(T)$, but the depth can get larger.
We show that the width of a partition of the edges as defined above is submodular.
We need this property to show that our main construction does not enlarge the width of the \preTreeDec.

\begin{lemma}
	For every graph $G$, $\wid$ is \submod.
\end{lemma}

\begin{proof}
	Let $P=\{X_1,\ldots,X_d\},Q=\{Y_1,\ldots,Y_d\}\in\partitions(E(G))$.
	We prove that \[\wid(P)+\wid(Q)\geq \wid(P_{X_1\ext \complementOf{Y_1}})+\wid(Q_{Y_1\ext \complementOf{X_1}}),\] which is enough to prove the lemma by symmetry.
	
	If $X_1=E(G)$, then $P=P_{X_1\ext \complementOf{Y_1}}$ and $Q=Q_{Y_1\ext \complementOf{X_1}}$, thus the lemma holds.
	
	If $X_1=\emptyset$, then $Q_{Y_1\ext\complementOf{X_1}}=(E(G),\emptyset,\ldots,\emptyset)$ and thus $\wid(Q_{Y_1\ext\complementOf{X_1}})=0$.
	Furthermore $\delta(P_{X_1 \ext \complementOf{Y_1}}) \subseteq \delta(P) \cup \delta(Y_1) \subseteq \delta(P) \cup \delta(Q)$ and thus $\wid(P_{X_1 \ext \complementOf{Y_1}}) \leq \wid(P) + \wid(Q)$, thus the lemma holds.
	
	If $Y_1=E(G)$ and $Y_1=\emptyset$ the lemma holds analogously.
	
	Thus let $\emptyset\neq X_1,Y_1\neq E(G)$.
	Trivially we get that $\delta(P_{X_1 \ext \complementOf{Y_1}}) \subseteq \delta(P) \cup \delta(Y_1)$ and $\delta(Q_{Y_1 \ext \complementOf{X_1}}) \subseteq \delta(Q) \cup \delta(X_1)$.
	Assume there exists some $v\in \delta(P_{X_1 \ext \complementOf{Y_1}})\setminus \delta(P)$, then $v\in \delta(Y_1)$ and thus $v\in \delta(Q)$.
	Furthermore we get that $E(v) \cap X_1=\emptyset$ and thus $E(v)\subseteq Y_1 \cup \complementOf{X_1}$.
	But then $v\notin \delta(Q_{Y_1 \ext \complementOf{X_1}})$.
	Analogously we can show that $\left(\delta(Q_{Y_1 \ext \complementOf{X_1}})\setminus \delta(Q)\right) \cap \delta(P_{X_1 \ext \complementOf{Y_1}})=\emptyset$.
	Thus all in all every vertex that is newly introduced to one of $\delta(P_{X_1 \ext \complementOf{Y_1}}), \delta(Q_{Y_1 \ext \complementOf{X_1}})$ is removed from the other and therefore the lemma holds.
\end{proof}

The next lemma shows that a \preTreeDec\ of a graph $G$ is indeed a relaxation of a tree decomposition of $G$.
If every edge is exact and all bags are exactly the boundary of the partition then we can construct a tree decomposition.
We need to start with a \preTreeDec\ of the graph $G^\circ$ with all self-loops added to ensure that every non-isolated vertex does appear in some bag and that the components correspondent to isolated vertices are covered by the \preTreeDec.
On the other hand we can transform a tree decomposition into a \preTreeDec, by copying the tree decomposition of each connected component of $G$ and adding leafs that correspond to the edges of $G^\circ$.

\begin{lemma}
	\label{lem:tw-ptw}
	Let $k,q\geq 1$.
	Let $G=(V,E)$ be a graph.
	Any tree decomposition of $G$ of width $\leq k-1$ and depth $\leq q$ gives rise to an exact \preTreeDec\ of $G^\circ$ of width $\leq k-1$ and depth $\leq q$ and vice versa.
\end{lemma}

\begin{proof}
	Let $(T,r,\beta,\gamma)$ be an exact \preTreeDec\ of $G^\circ$ of width $\leq k$ and depth $\leq q$.
	We define $\beta'\colon V(T)\rightarrow 2^{V(G)}$ as follows
	\begin{equation*}
	\beta'(t)\coloneqq\begin{cases}
	\{v\} & \text{ if } t\in L(T), r \text{ parent of } t \text{ and } \gamma(r,t)=\{vv\},\\
	\beta(t) & \text{ otherwise.}
	\end{cases}
	\end{equation*}
	\begin{claim}
		$(T,\beta')$ is a tree decomposition of width $\leq k$ and depth $\leq q$.
	\end{claim}
	\begin{claimproof}
		From \ref{ax:preTreeRoot}, \ref{ax:preTreeLeaf} and \cref{obs:exact-subtree} we get that for every edge $uv\in E(G^\circ)$ there is some leaf $\ell$ with parent $p$ and $\gamma(p,\ell)=\{uv\}$.
		Thus if $u=v$, then $\beta'(t)=\{v\}$ and thus $vv\in E(G[\beta'(t)])$.
		Otherwise $uu,vv\in E(G^\circ)\setminus \{uv\}$ and thus $u,v\in\beta'(t)$ and $uv\in E(G[\beta'(t)])$.
		All in all we get that \ref{ax:TreeDec1} holds.
		
		Assume there exists a $v\in V(G)$ such that $T_v$ is not connected.
		Let $T_1,T_2$ be two disjoint connected components of $T_v$ and let $P=t_1,\ldots,t_a$ be the shortest $T_1$-$T_2$-path in $T$.
		Then $v\notin \delta(\gamma(t_1,t_2))\subseteq \delta(\pi_{t_2})$ and thus $E(v)\cap \gamma(t_1,t_2)=\emptyset$.
		As all edges in $P$ are exact it holds that $\gamma(t_1,t_2)\supseteq \gamma(t_a,s)$, for all $s\in N(t_a) \setminus\{t_{a-1}\}$.
		And thus it holds that $E(v)\cap \gamma(t_a,s)=\emptyset$ and $E(v)\subseteq \gamma(t_a,t_{a-1})$.
		This contradicts $v\in\bigcup_{s\in N(t_a)} \delta(\gamma(t_a,s))$.
		Therefore \ref{ax:TreeDec2} also holds and $(T,\beta')$ is a tree decomposition.
		
		The width and depth are obvious as $k,q\geq 1$.
	\end{claimproof}
	
	Now let $(T,r,\beta)$ be a tree decomposition of $G$ of width $\leq k$ and depth $\leq q$.
	W.l.o.g. $\beta$ is \emph{tight}, that is for all $t\in V(T)$ and $v\in \beta(t)$, that $(T,r,\beta')$, where $\beta'(t)\coloneqq\beta(t)\setminus\{v\}$ and $\beta'(s)=\beta(s)$, for all $s\in V(T)\setminus \{t\}$, is not a tree decomposition of $G$.
	We construct a new tree $T'$ with root $r'$ and functions $\beta'\colon V(T')\rightarrow 2^{V(G)}$, $\gamma\colon \overrightarrow{E(T')} \rightarrow 2^{E(G^\circ)}$, $f\colon V(T')\setminus\left(L(T')\cup\{r'\}\right)\rightarrow V(T)$ as follows.
	Let $C$ be a connected component of $G$ and let $V_C\coloneqq \{t\in V(T)\mid V(C)\cap \beta(t)\neq \emptyset\}$.
	By \cref{lem:tdec-conn} $V_C$ is connected.
	If $C$ is an isolated vertex $v$, then $V_C=\{t\}$, for some $t\in V(T)$.
	We add a new node $t_v$ to $T'$ and connect it to the root.
	We set $\beta'(t_v)=\emptyset$, $\gamma(r',t_v)=\{vv\}$ and $\gamma(t_v,r')= E(G^\circ)\setminus \{vv\}$.
	Otherwise let $T_C$ be a copy of the subtree induced by $V_C$ with root $r_C$ and vertices $V^*_C$ and $f|_{V^*_C}\colon V^*_C\rightarrow V_C$ the natural bijection between the copies and their originals.
	We attach $r_C$ to the root $r'$.
	For every $v\in V(C)$, there is some $t_v\in V_C$ such that $v\in\beta(t_v)$, as $C$ is not an isolated vertex.
	We add a new leaf $t'_v$ that we attach to $f|_{V(T_C)}^{-1}(t_v)$ and set $\beta'(t'_v)=\{v\}$, $\gamma(f|_{V^*_C}^{-1}(t_v),t'_v)=\{vv\}$ and $\gamma(t'_v,f|_{V^*_C}^{-1}(t_v))=E(G^\circ)\setminus \{vv\}$.
	For every $e\in E_G(C)$ there is some $t_e\in V_C$ such that $e\subseteq\beta(t_e)$.
	We add a new leaf $t'_e$ that we attach to $f|_{V^*_C}^{-1}(t_e)$ and set $\beta'(t'_e)=e$, $\gamma(f|_{V^*_C}^{-1}(t_e),t'_e)=\{e\}$ and $\gamma(t'_e,f|_{V^*_C}^{-1}(t_e))=E(G^\circ)\setminus \{e\}$.
	For every node $t\in V^*_C$ with parent $p$ we add all edges $e\in E_G(C)$, where $t'_e$ is a descendant of $t$, and all self-loops $vv\in E_{G^\circ}(C)$, where $t'_v$ is a descendant of $t$, to $\gamma(p,t)$.
	Furthermore we set $\gamma(t,p)\coloneqq E(G^\circ)\setminus \gamma(p,t)$ and $\beta'(t)\coloneqq \delta(\pi_t)\subseteq\beta(f(t))\cap V(C)$.
	By tightness of $\beta$ there is some $v\in \beta(f(\ell))$ such that $T_v=\{f(\ell)\}$, for every $\ell\in L(T_C)$, thus no leaf of $T_C$ is a leaf in $T'$, thus $(T',r',\beta',\gamma)$ satisfies \ref{ax:preTreeLeaf}.
	\ref{ax:preTreeRoot}, \ref{ax:preTreePart} and \ref{ax:preTreeEdge} hold by construction.
	Furthermore every edge is exact by construction.
	Thus $(T',r',\beta',\gamma)$ is an exact \preTreeDec\ of $G^\circ$.
	
	The width is obvious as every bag in $\beta'$ is a subset of some bag in $\beta$.
	To see that the depth bound also holds we observe two things.
	For every leaf $\ell\in L(T')$ with parent $p$ we get that $\beta'(\ell)\setminus\beta'(p)=\emptyset$.
	For every inner node $t\in V(T')\setminus L(T')$ with parent $p$ we get that $\beta'(t)\setminus\beta'(p)\subseteq \beta(f(t))$ and, if $p\neq r'$, $\beta'(t)\setminus\beta'(p)\subseteq \beta(f(t))\setminus \beta(f(p))$, by the tightness of $\beta$.
\end{proof}

The next observation can be seen by following the same arguments as in the proof of \cref{lem:tw-ptw}.

\begin{observation}
	\label{obs:exact-subtree-depth}
	Let $(T,r,\beta,\gamma)$ be a pre-decomposition of some graph $G$.
	Let $(T',r)$ be a subtree of $(T,r)$ with the same root.
	If all edges in $T'$ are exact then, for every $v\in V(G)$, the set 
	$\{t\in V(T')\mid v\in \delta(\pi_s)\}$ is connected in $T'$.
	In particular every $t\in V(T')$ satisfies
	\begin{equation*}
		\sum_{s\in P_t\setminus \{r\}} |\delta(\pi_s) \setminus \delta(\pi_{p_s})|=|\bigcup_{s\in P_t} \delta(\pi_s)|.
	\end{equation*}
\end{observation}

We conclude this section with an observation about the cones along a path of exact edges.
It is a direct consequence of exactness and the fact that the cones incident to a vertex form a partition of the edges.

\begin{observation}
	\label{obs:exact-path}
	Let $(T,r,\beta,\gamma)$ be a pre-decomposition of some graph $G$.
	Let $P=t_1,\ldots,t_\ell$ a path in $T$, such that every edge $t_it_{i+1}$, for $i\in[\ell-1]$, is exact.
	Then it holds that
	$\gamma(t_1,t_2) \supseteq \gamma(t_2,t_3) \supseteq \ldots \supseteq \gamma(t_{\ell-1},t_\ell)$.
\end{observation}

\section{The game}\label{sec:game}

In the cops-and-robber game on a graph $G$, the cops occupy sets $X$ of at most $k$ vertices of $G$, and the robber moves on edges of $G$. 
In order to make the rules precise, we need \emph{edge components} of $G$ that arise when the cops are blocking a set $X$.

\begin{definition}
	Let $G=(V,E)$ be a graph and $X\subseteq V$. We let \emph{the edge component graph} of $G$ with respect to $X$ be the graph $G^{X}$ obtained as the disjoint union of the following graphs. (In order to make all graphs disjoint we introduce copies of vertices where needed.)
	\begin{itemize}
		\item For every $uv\in E(G[X])$, the graph $G_{uv}\coloneqq(\{u,v\},\{uv\})$, and
		\item for every connected component $C$ of $G\setminus X$, the graph $G_C$, with 
			$V(G_C)\coloneqq V(C)\cup N_G(V(C))$ and $E(G_C)\coloneqq  E'(C)\cup E(V(C),X)$, 
			where $E(V(C),X)$ is the set of edges of $G$ incident to both a vertex of $C$ and a vertex in $X$.
	\end{itemize}
\end{definition}

The reader may note that $G^X$ may contain multiple copies of the vertices in $X$, but exactly one copy of each edge in $G$.

\begin{observation}
	There is a natural bijection $\Psi\colon E(G^X)\rightarrow E(G)$ between the edges of $G^X$ and the edges of $G$.
\end{observation}

\begin{definition}[$q$-placement $k$-cops-and-robber game]
\label{def:Ekq-cops}
Let $G$ be a graph and let $k,q\geq 1$.
The \emph{$q$-placement $k$-cops-and-robber game} $\CRkq$ is defined as follows:

We have a counter $j$ that indicates how many times cops are placed on the graph.
If $G$ does not contain any edges the cop player wins immediately.

\begin{itemize}
	\item We initialize the counter $j=0$.
	\item The cop positions are sets $X\in V(G)^{\leq k}$.
	\item The robber position is an edge $uv\in E(G)$.
	\item The initial position $(X_0,u_0v_0)$ of the game is $X_0 = \emptyset$ and $u_0v_0\in E(G)$, thus the game starts with no cops positioned on $G$ and the robber on an arbitrary edge in a connected component of $G$ of his choice.
	\item For $X\subseteq V(G)$ and $uv\in E(G)$, we write $\gamma^X_{uv}\coloneqq\Psi(E(C))$ for the component $C$ of the graph $G^X$, such that $uv\in E(\gamma^X_{uv})$.
	Thus if the cops are at positions $X$ and robber at an edge $uv$ we write $(X,\gamma^X_{uv})$ for the position of the game.
	\item In round $i$ the cop-player can move from the set $X_{i-1}$ to a set $X_{i}$, if $X_{i}\subseteq X_{i-1}$,
		or if $j<q$, the cop-player can pick a vertex $v\in V(G)$ and play to $X_{i}\coloneqq X_{i-1}\cup\{v\}$ and increase $j$ by one.
	\item In round $i$ the robber-player can move along a path with no internal vertex in $X_{i-1}\cap X_{i}$. 
		Thus the robber-player can move to some edge $u_iv_i$, such that the edge $\Psi^{-1}(u_i v_i)$ is in a connected component of $G^{X_i}$ that is contained in $\Psi^{-1}(\gamma^{X_{i-1}\cap X_i}_{u_{i-1}v_{i-1}})$ via a path $p=w_1,\ldots,w_\ell$ where $\{w_1,w_2\}=\{u_{i-1},v_{i-1}\}$ and $\{w_{\ell-1},w_\ell\}=\{u_i,v_i\}$ and $\{w_2,\ldots,w_{\ell-1}\}\cap X_i \cap X_{i-1}=\emptyset$.
	\item The cop-player wins in round $i$, if $\{u_i,v_i\}\subseteq X_i$, and we say the cop-player \emph{captures} the robber in round $i$. The robber-player wins if the cop-player has not won and $j=q$.
\end{itemize}

	If we further restrict the movement of the cops such that always $\gamma_{u_{i-1}v_{i-1}}^{X_{i-1}}\supseteq \gamma_{u_{i-1}v_{i-1}}^{X_{i-1}\cap X_{i}}$ holds, we write \monCRkq\ and call the game \emph{monotone $q$-placement $k$-cops-and-robber game}.
\end{definition}

The game played on the graph $G^\circ$ corresponds to the game on $G$, where the robber can hide both inside a vertex or an edge.
It is easy to see that this does not benefit the robber player, that is he wins the game $\CRkq$ if and only if he wins the game $\CRkq[G^\circ]$, as the components that are reachable by the robber are essentially the same.
In \cite{FluckSS24}, the authors introduce a cops-and-robber game, where the robber can only hide in the vertices.
Again this does not pose a restriction for the robber with the same argument as above.
There is a tight connection between the cops-and-robber game defined above and tree decompositions of graphs.

\begin{lemma}[\cite{FluckSS24}]
	\label{lem:moncr-vs-tkq}
	Let $G$ be a graph and $k,q\in\NN$.
	The cop player wins $\monCRkq$ if and only if $G\in\mathcal{T}^k_q$.
\end{lemma}

Towards strengthening the above connection to also include the non-monotone game we first introduce how to construct a \preTreeDec\ from a winning strategy of the cop player.

\begin{definition}[strategy tree]
\label{def:strat-tree}
Let $G$ be a graph without isolated vertices and let $k,q\in\NN$.
Let $\sigma\colon V(G)^{\leq k}\times E(G)\rightarrow V(G)^{\leq k}$ a cop strategy such that that for all $X\in V(G)^{\leq k}$, for all $uv\in E(G)$ and for all $u'v'\in\gamma_{uv}^X$ we have that $\sigma(X,uv)=\sigma(X,u'v')$.
We write $\sigma(X,\gamma_{uv}^X)$ instead of $\sigma(X,uv)$.

The \emph{strategy tree of $\sigma$} is a \preTreeDec\ $(T,r,\beta,\gamma)$, inductively defined as follows:
\begin{itemize}
	\item $\beta(r)=\emptyset$,
	\item for every connected component $C$ of $G$, there is a child $c$ of the root $r$ and $\gamma(r,c)=E(C)$,
	\item for every node $t\in V(T)\setminus\{r\}$ with parent $s\in V(T)$,
	\begin{itemize}
		\item if the robber player is caught, we set $\beta(t)= e$, where $\gamma(s,t)=\{e\}$,
		\item else $\beta(t)=\sigma(\beta(s),\gamma(s,t))$ and
		\item for every connected component $C$ of $G^{\beta(t)}$, that has a non-empty intersection with $\Psi^{-1}(\gamma(s,t))$, there is a child $c$ of $t$ and $\gamma(t,c)=\Psi(E(C))$,
			
		\item $\gamma(t,s)\coloneqq E(G)\setminus \bigcup_{c \text{ child of }t}\gamma(t,c)$, if $t\notin L(T)$, and
		\item $\gamma(t,s)\coloneqq E(G)\setminus\gamma(s,t)$, if $t\in L(T)$.
	\end{itemize}
\end{itemize}
We call $t\in V(T)$ a \emph{branching node} if the cop player placed a new cop incident to the robber escape space.

Observe that if $t\in V(T)$ is a leaf, then the robber is captured and the depth of $(T,r,\beta,\gamma)$ is $\leq q$ if and only if $\sigma$ is a winning strategy in \CRkq.
\end{definition}

Note that w.l.o.g. every child of the root is a branching node, as the cop player w.l.o.g. only plays positions that are inside the component the robber chose in the first round.
If the game is played on $G^\circ$, then every branching node that does not correspond to the placement of a cop onto an isolated vertex has more than one child.
We observe that the monotone moves of the cop player correspond to the exact edges in the strategy tree.

\begin{observation}
	\label{obs:monotone-exact}
	For edge $st \in E(T)$, where $s$ is the parent of $t$ it holds that the move $\sigma(\beta(s),\gamma(s,t))$ is monotone if and only if $st$ is exact.
	Moreover, if $st$ is not exact, then $\beta(t) \subsetneq \beta(s)$.
\end{observation}

The following two observations about the self-loops in the graph $G^\circ$ are key to prove the construction in the next section does not enlarge the depth of the \preTreeDec.

\begin{observation}
	\label{obs:self-loops-gamma}
	When considering the game on $G^\circ$, all self-loops $vv$ incident to $\beta(s)$ are either contained in $\gamma(s,p)$ or there is a child $c$ of $s$ such that $\gamma(s,c)=\{vv\}$.
\end{observation}

\begin{observation}
	\label{obs:self-loops}
	Let $v\in V(G)$ be a non-isolated vertex and the game played on $G^\circ$.
	A node $s$ has a child $c$ with $\gamma(s,c)=\{vv\}$, for some self-loop $vv$ if and only if $s$ is a branching node and $v$ is the vertex the cops picked.
\end{observation}

\section{Making a strategy tree exact}\label{sec:make-exact}

Our goal is to prove the following theorem.

\begin{theorem}
	\label{thm:exactness}
	Let $G=(V(G),E(G))$ be a graph, let $k,q\geq 1$ and let $(T,r,\beta,\gamma)$ be a strategy tree for some cop strategy $\sigma\colon V(G^\circ)^k\times E(G^\circ)\rightarrow V(G^\circ)^k$.
	If $\sigma$ is a winning strategy in $\CRkq[G^\circ]$, then there is a tree decomposition of $G$ with width $\leq k$ and depth $\leq q$.
\end{theorem}

To prove this we construct an exact \preTreeDec\ of $G^\circ$ from the strategy tree, starting at the root $r$ and traversing the tree nodes in a breadth-first-search.
We then use \cref{lem:tw-ptw} to get the desired tree decomposition.
When we \consider\ a node we change the \preTreeDec\ so that all incident edges are exact afterwards.
Note that by the choice of the traversal we only need to consider outgoing edges.

\subparagraph*{The construction.}
Let $(T,r,\beta,\gamma)$ be the \preTreeDec\ of $G^\circ$ from a winning strategy as in \cref{thm:exactness}.
Let $s_1,\ldots,s_{n_T}$ be an order of the nodes of $T$ in bfs where $s_1=r$. Let $\beta_0\coloneqq\beta$ and $\gamma_0\coloneqq \gamma$.
We construct a sequence $(T,r,\beta_0,\gamma_0),\ldots,(T,r,\beta_{n_T},\gamma_{n_T})$ of \preTreeDec s, such that $(T,r,\beta_{n_T},\gamma_{n_T})$ is exact.
We say $s_i$ is \emph{\considered\ in step $i$}.
Let \[T_i\coloneqq T[\{s_1,\ldots,s_i\}\cup N_T(\{s_1,\ldots,s_i\})].\]
See \cref{fig:teeEye} for an illustration of $T_i$. (It will become clear that this is the subtree of all nodes where the \preTreeDec\ is modified in or before step $i$. We also point out that edges from $T_i$ to $T\setminus T_i$ may become non-exact during our modification process.)

If $s_i$ is a leaf, there are no outgoing edges that are not exact, and we set $\beta_i\coloneqq\beta_{i-1}$ and $\gamma_i\coloneqq \gamma_{i-1}$.
Otherwise let $t^i_1,\ldots,t^i_{a_i}\in N_T(s_i)$ be all children of $s_i$.

\begin{figure}
    \centering
    \begin{tikzpicture}
        \node at (-1,0) {$T_i =$};
    
        \coordinate (A) at (1.5,1);
        \coordinate (B) at (0,-.25);
        \coordinate (C) at (3,-.25);
        
        \draw (A) -- (B) -- (C) -- (A);
        
        \coordinate (si) at (2,0);
        \node (Ti1) at (1.5,0.5) {$T_{i-1}$};
        \coordinate (t1) at (1,-.75);
        \coordinate (ti) at (3,-.75);
        \coordinate (dots1) at (2,-.75);
        \coordinate (dots2) at (2,-1);

        \fill [black] (si) circle (2pt) node[above] {$s_i$};
        \fill [black] (t1) circle (2pt) node[below] {$t_{1}^{i}$};
        \fill [black] (ti) circle (2pt) node[below] {$t_{a_i}^{i}$};
        \node at (dots1) {$\dots$};
        \node at (dots2) {$\dots$};

        \draw (si)--(t1);
        \draw (si)--(ti);
    \end{tikzpicture}
	\caption{The subtree $T_i$ appearing in the construction.}  
\label{fig:teeEye}
\end{figure}
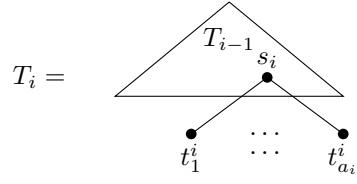

\begin{itemize}
	\item We pick pairwise disjoint $F^i_1,\ldots,F^i_{a_i}\subseteq E(G^\circ)$, with \[F^i_j\subseteq \complementOf{\gamma_{i-1}(t^i_j,s_i)}\cap\complementOf{ \gamma_{i-1}(s_i,t^i_j)},\] such that the partition $\pi^*$ that results from taking the $F^i_j$-extensions in $\gamma_{i-1}(s_i,t^i_j)$ (in arbitrary order) has the minimum size boundary.
	If there are multiple optimal choices for $F^i_1,\ldots,F^i_{a_i}$ we select the one that minimizes the size of $\bigcup_{j\in[a_i]} F^i_j$, if there are still several options we break ties arbitrarily.
	\item Let $F^i\coloneqq\bigcup_{j\in[a_i]} F^i_j$ and $F^{*i}_j\coloneqq \left(\left(\complementOf{\gamma_{i-1}(t^i_j,s_i)}\cap \complementOf{\gamma_{i-1}(s_i,t^i_j)}\right)\cup F^i\right) \setminus F^i_j$.
\end{itemize}

For every $p\in V(T_i)$ with child $c$ we set
\begin{equation*}
\gamma_i(p,c)\coloneqq\begin{cases}
	\left(\gamma_{i-1}(s_i,t^i_j)\setminus F^i\right)\cup F^i_j & \text{ if } (p,c)=(s_i,t^i_j), \text{ for some } j\in[a_i],\\
	\gamma_{i-1}(p,c)\setminus F^{*i}_j & \text{ if } t^i_j= p, \text{ for some } j\in[a_i],\\
	\gamma_{i-1}(p,c)\cup F^i & \text{ if } c\preceq s_i,\\
	\gamma_{i-1}(p,c)\setminus F^i & \text{ otherwise,}
\end{cases}
\end{equation*}
and
\begin{equation*}
\gamma_i(c,p)\coloneqq\begin{cases}
	\gamma_{i-1}(c,p) \cup F^{*i}_j & \text{ if } (p,c)=(s_i,t^i_j), \text{ for some } j\in[a_i],\\
	\gamma_{i-1}(c,p) & \text{ if } t^i_j= p, \text{ for some } j\in[a_i],\\
	\gamma_{i-1}(c,p) \setminus F^i & \text{ if } c\preceq s_i,\\
	\gamma_{i-1}(c,p) \cup F^i & \text{ otherwise,}
\end{cases}
\end{equation*}
and all other $uv\in\overrightarrow{E(T)}$ we set $\gamma_i(u,v)\coloneqq\gamma_{i-1}(u,v)$.
Furthermore we set
\begin{equation*}
\beta_i(t)\coloneqq\begin{cases}
	\delta(\pi^i_t) & \text{ if } t\in V(T_i),\\
	\beta_{i-1}(t) & \text{ otherwise.}
\end{cases}
\end{equation*}

Intuitively in the construction above we push the change at $s_i$ through $T_{i-1}$, that is for all edges in $T_{i-1}$ we add $F^i$ to the directed edge that points away from $s_i$ and remove $F^i$ from the edges in the other direction.
We obtain the following observation.

\begin{observation}
	\label{obs:only-remove}
	Let $i,j\in[n_T]$ such that $s_i$ is the parent of $s_j$.
	Then $\gamma_{\alpha}(s_i,s_j) \subseteq \gamma(s_i,s_j)$, for all $\alpha<i$.
\end{observation}

\subparagraph*{The proof.}

We prove \cref{thm:exactness} in three steps.
First we prove that the construction indeed yields an exact pre-tree decomposition. 
Next we show that the width can be bounded as desired and lastly we prove that the construction yields the desired depth.

\begin{lemma}
	\label{lem:exact}
	For all $i\in[n_T]$, $(T,r,\beta_i,\gamma_i)$ is a pre-tree decomposition.
	Furthermore all edges in $E(T_i)$ are exact.
\end{lemma}

\begin{proof}
	\ref{ax:preTreeRoot} holds as all edges leaving the root are already exact in $\gamma$, thus we change nothing in step 1 where the root is considered and every $F^i$, with $i>1$, only contains edges from a single component of $G$ by construction.
	
	We observe that the changes from $\gamma_{i-1}$ to $\gamma_{i}$ at some node $t\in V(T_i)\setminus \{s_i\}$ corresponds to an $F^i$- or $F^{*i}_j$-extension of $\pi^{i-1}_t$ at the set that corresponds to the edge, that points towards $s_i$.
	Furthermore $\pi^i_s$ is a partition of the edges by construction and for all $t\in V(T_i)$ we set $\beta_i(t)=\delta(\pi^i_t)$.
	As $\gamma_i$ and $\beta_i$ are equal to $\gamma$ and $\beta$ at all vertices that are not part of $V(T_i)$, this shows by induction that \ref{ax:preTreePart} still holds.
	
	Next we observe that at every edge that is not incident to some $t^i_j$ we add to one direction exactly what we remove from the other direction.
	Furthermore by construction the edges $s_i t^i_j$ are exact after the construction.
	Lastly, for all children $c$ of $t^i_j$, we only remove edges from $\gamma_i(t^i_j,c)$.
	Thus again by induction we get that \ref{ax:preTreeEdge} holds and that all edges of $T_i$ are exact.
	
	It remains to show that \ref{ax:preTreeLeaf} holds.
	Let $i\in[n_t]$ and $j\in[a_i]$ such that $t^i_j\in L(T)$.
	From \cref{obs:only-remove} we know that $\gamma_{i-1}(s_i,t^i_j)\subseteq \gamma(s_i,t^i_j)$ and thus $|\gamma_{i-1}(s_i,t^i_j)|\leq 1$.
	Furthermore we know that $\gamma_{i-1}(t^i_j,s_i)=\gamma(t^i_j,s_i)= E(G^\circ)\setminus \gamma(s_i,t^i_j)$.
	Therefore we get that $F^i_j\subseteq \gamma(s_i,t^i_j)$ and thus $\gamma_{i}(s_i,t^i_j)\subseteq \gamma(s_i,t^i_j)$.
	By construction, in step $i'\in[n_T]$, such that $s_{i'}=t^i_j$, we do nothing.
	And in all other steps $\alpha > i$, we have that $t^i_j\not\preceq s_\alpha$ and thus we only remove edges from $\gamma_{\alpha}(s_i,t^i_j)$.
	This shows that for all $\alpha\in[n_T]$ we have $|\gamma_{\alpha}(s_i,t^i_j)|\leq 1$.
\end{proof}

Hence, for $i=n_T$, we get that $(T,r,\beta_{n_T},\gamma_{n_T})$ is an exact \preTreeDec. 
Note that it is possible that $\gamma_{n_T}(s,t)$ is empty for an edge $st\in \overrightarrow{E(T)}$.
By Lemma~\ref{lem:tw-ptw} we obtain a tree decomposition, from this \preTreeDec.
We show below that the width and depth are as stated in the theorem.

Our construction does not change the width of the decomposition.
To prove this we observe that in step $i$ the bound in $s_i$ is minimal.
We then push the change through the subtree $T_i$ and find that if a change would increase the width, we could push this change back to the node $s_i$ and find an even smaller bound there, which contradicts the minimality of our choice.

\begin{lemma}
	\label{lem:width}
	$\wid(T,r,\beta_i,\gamma_i)\leq \wid(T,r,\beta,\gamma)$, for all $i\in[n_T]$.
\end{lemma}

\begin{proof}
	We prove the statement for all $0\leq i\leq n_T$ by induction.
	As $(T,r,\beta_0,\gamma_0)=(T,r,\beta,\gamma)$, the statement clearly holds for $i=0$.
	Next we show that $\wid(T,r,\beta_i,\gamma_i)\leq \wid(T,r,\beta_{i-1},\gamma_{i-1})$, for all $i\in[n_T]$.
	Obviously $|\beta_i(t)|=\beta_{i-1}(t)$, for all $t\notin T_i$.
	Furthermore by construction $|\beta_i(s_i)|\leq |\beta_{i-1}(s_i)|$.
	Let $j\in[a_i]$, let $X\coloneqq \gamma_i(s_i,t^i_j)$ and let $Y\coloneqq\gamma_{i-1}(t^i_j,s_i)$.
	We observe that \[\pi^{i}_{t^i_j}=\pi^{i-1}_{t^i_j,Y \ext \complementOf{X}}.\]
	Thus it holds that \[|\beta_i(t^i_j)| = |\wid(\pi^{i-1}_{t^i_j, Y \ext \complementOf{X}})| \leq |\wid(\pi^{i-1}_{t^i_j})| \leq |\beta_{i-1}(t^i_j)|\] as otherwise by submodularity for the partitions $\pi^{i-1}_{t^i_j}$ and $\pi^i_{s_i}$, we get that \[|\wid(\pi^i_{s_i, X \ext \complementOf{Y}})| < |\wid(\pi^i_{s_i})|,\]
	which contradicts the minimality of the bound for $F^i_1,\ldots,F^i_{a_i}$.
	
	Lastly assume there is a node $t$ in $V(T_i)\setminus\{s_i,t^i_1,\ldots,t^i_{a_i}\}$ such that $|\beta_i(t)|>|\beta_{i-1}(t)|$.
	We assume $t$ is of minimal distance to $s_i$ with this property.
	Let $x_0=t,x_1,\ldots,x_{b}=s_i$ be the path from $t$ to $s_i$.
	By minimality of the distance we know that $|\beta_i(x_1)|\leq |\beta_{i-1}(x_1)|$.
	Additionally we know that all edges on the path from $s_i$ to $x_1$ are exact in $\gamma_i$, as well as the edge $x_1t$ in $\gamma_{i-1}$.
	Now let $Y\coloneqq \gamma_{i-1}(t,x_1)$ and, for all $0\leq \alpha < b$, let $X_\alpha \coloneqq \gamma_i(x_{\alpha + 1},x_{\alpha})$ and $Z_\alpha \coloneqq \gamma_i(x_{\alpha},x_{\alpha + 1})$.
	The transition from $i-1$ to $i$ at $t$ corresponds to $\pi^{i-1}_{t,Y \ext F}=\pi^{i-1}_{t,Y \ext \complementOf{X_0}}$.
	Thus if $\wid(\pi^{i-1}_{t,Y \ext F})=|\beta_i(t)|>|\beta_{i-1}(t)|=\wid(\pi^{i-1}_{t})$ we get by submodularity that $\wid(\pi^i_{x_1})>\wid(\pi^i_{x_1,X_0 \ext \complementOf{Y}})$.
	As the edge $x_1t$ was exact at step $i-1$, we know that  \[F'\coloneqq\complementOf{Y}\setminus X_0= F^i\setminus Y \subseteq F^i.\]
	We now push this change back to $s_i$ along the path $x_1,\ldots,x_b$ and we again find a contradiction to the minimality of the bound of $F^i_1,\ldots,F^i_{a_i}$.
	For this, let us assume we have pushed the change to $x_\alpha$, that is we changed $\pi^i_{x_\alpha}$ to $\pi^*_{x_\alpha} = \pi^i_{x_\alpha,X_{\alpha - 1}\ext F'}$ and we know that $\wid(\pi^*_{x_\alpha})<\wid(\pi^{i}_{x_\alpha})$.
	As the edge $x_\alpha x_{\alpha+1}$ is exact in $\gamma_i$, we get that $\pi^*_{x_\alpha,(Z_\alpha \setminus F') \ext \complementOf{X_\alpha}} = \pi^i_{x_\alpha}$.
	Let \[\pi^*_{x_{\alpha+1}}\coloneqq \pi^i_{x_{\alpha+1},X_\alpha \ext \complementOfB{(Z_\alpha \setminus F')}} = \pi^i_{x_{\alpha+1},X_\alpha \ext F'},\] then by submodularity $\wid(\pi^*_{x_{\alpha+1}}) < \wid(\pi^i_{x_{\alpha+1}})$.
	When we have pushed the change to $\alpha=b$, we find the desired contradiction.
\end{proof}

To prove that our construction does not increase the depth we show that in every step $i$ the depth up to the nodes in $T_i$ is bounded by the depth up to these nodes in the original tree.
We prove this by induction on the number of steps.
In step $i$ every change in any bag at some node in $V(T_{i-1})$ is closely related to the change at the \considered\ node $s_i$.
Additionally we find that the vertices in a bag at some child of $s_i$ that is not present in the bag at $s_i$ is exactly the vertex the cop player newly placed in the corresponding move of the game.

\begin{lemma}
	\label{lem:depth}
	For all $i\in[n_T]$ and all $t\in V(T_i)$, it holds that
	\begin{equation*}
	\sum_{s\in P_t\setminus\{r\}} |\beta_i(s) \setminus \beta_i(p_s)|\leq \sum_{s\in P_t\setminus\{r\}} |\beta(s) \setminus \beta(p_s)|.
	\end{equation*}
\end{lemma}

\begin{proof}
	Let $\ell\in[n_T]$.
	As by construction $\beta_\ell(t)=\delta(\pi^\ell_t)$, for all $t\in V(T_\ell)$, we get from \cref{obs:exact-subtree-depth} and \cref{lem:exact} that $|\bigcup_{s\in P_t} \beta_\ell(s)| = \sum_{s\in P_t\setminus\{r\}} |\beta_\ell(s) \setminus \beta_\ell(p_s)|$.
	Thus it suffices to show that $|\bigcup_{s\in P_t} \beta_i(s)| \leq \sum_{s\in P_t\setminus\{r\}} |\beta(s) \setminus \beta(p_s)|.$
	
	We prove the statement by induction on the steps $i$.
	Recall that $(T,r,\beta_0,\gamma_0)=(T,r,\beta,\gamma)$, thus the statement holds for $i=0$.
	Now assume the statement holds for $i-1$, thus for all $t\in V(T_{i-1})$ it hold that  $|\bigcup_{s\in P_t} \beta_{i-1}(s)| \leq \sum_{s\in P_t\setminus\{r\}} |\beta(s) \setminus \beta(p_s)|.$
	
	We recall that $V(T_i)=V(T_{i-1})\cup \{t^i_1,\ldots,t^i_{a_i}\}$ and that $s_i\in L(T_{i-1})$.
	For the nodes $t\in V(T_{i-1})$ we can directly build upon the induction hypothesis.
	But the nodes $t^i_j$, with $j\in [a_i]$, are added into the subtree.
	Here we need to compare directly to the original bags, as we can no longer use that in step $i-1$ the depth at these nodes is bounded by the depth in the original strategy tree.
	We can prove for these nodes that every vertex newly placed at one of these nodes in step $i$ is also newly placed in the original strategy.
	Then we can show that the difference between depth at these nodes and their parent in step $i$ can be bounded by the difference in the original strategy tree.
	\begin{claim}
		\label{cl:depth-leafs-new}
		 Every $j\in[a_i]$ satisfies $\beta_i(t^i_j)\setminus \beta_i(s_i)\subseteq \beta(t^i_j)\setminus \beta(s_i).$
	\end{claim}
	\begin{claimproof}
		Let $v \in \beta_i(t^i_j) \setminus \beta_i(s_i)$.
		As $v \notin \beta_i(s_i)$ we get that $v\notin\delta(\gamma_i(t^i_j,s_i))$ and thus $E_{G^\circ}(v)\cap \gamma_i(t^i_j,s_i) = \emptyset$.
		By construction we have that $\gamma_i(t^i_j,s_i)\supseteq {i-1}(t^i_j,s_i) = \gamma(t^i_j,s_i)$, and thus $vv \notin \gamma(t^i_j,s_i)$.
		As $v \in \beta_i(t^i_j)=\delta(\pi^i_{t^i_j})$, there are two distinct children $c_1,c_2$ of $t^i_j$ such that $v\in\delta(\gamma_i(t^i_j,c_\ell))$ and thus $E_{G^\circ}(v)\cap \gamma_i(t^i_j,c_\ell) \neq \emptyset$, for $\ell = 1,2$.
		By construction we have $\gamma_i(t^i_j,c_\ell)\subseteq \gamma_{i-1}(t^i_j,c_\ell) = \gamma(t^i_j,c_\ell)$, for $\ell = 1,2$.
		And thus $v\in \delta(\pi_{t^i_j})\subseteq \beta(t^i_j)$.
		By \cref{obs:self-loops-gamma} there thus is a child $c$ of $t^i_j$ such that $\gamma(t^i_j,c)=\{vv\}$ and, by \cref{obs:self-loops}, $v\in \beta(t^i_j)\setminus \beta(s_i)$.
	\end{claimproof}
	
	The following claim tracks vertices that are added to any bag in $V(T_{i-1})$ at step $i$.
	
	\begin{claim}
		\label{cl:new-vertex-new}
		Let $i\in [n_T]$ and let $t\in V(T_{i-1})$.
		If $v\in \beta_i(t)\setminus\beta_{i-1}(t)$, then $v\in\beta_i(t^*)$, for all $t^*$ on the path from $t$ to $s_i$.
	\end{claim}
	
	\begin{claimproof}
		Let $t^*\neq t$.
		Let $t'$ be the next node on the path from $t$ to $s_i$.
		Then $\gamma_i(t,t')=\gamma_{i-1}(t,t')\cup F^i$.
		As $\gamma_i(t,t')$ is the only set incident to $t$ where edges are added in step $i$, we get that $v\in\delta(\gamma_i(t,t'))$.
		And from $v\notin\delta(\gamma_{i-1}(t,t'))$ we get that $v\in\delta(F^i)$.
		Now suppose that $v\notin \beta_i(t^*)$, and thus also $v\notin \delta(\gamma_i(t^*,p))$, where $p$ is the next node on the path from $t^*$ to $t$.
		As $v$ is incident to edges in $F^i$ we get that $E_{G^\circ}(v) \cap \gamma_i(t^*,p) = \emptyset$.
		We know from \cref{lem:exact} that all edges in $T_i$ are exact and thus that $\gamma_i(t',t)\subseteq \gamma_i(t^*,p)$ by \cref{obs:exact-path}.
		This is a contradiction to $v\in\delta(\gamma_i(t,t'))=\delta(\gamma_i(t',t))$ and thus $v\in\beta_i(t^*)$.
	\end{claimproof}
	
	The next claim is used to show that a vertex that disappears from a bag in $V(T_{i-1})$ at step $i$ also disappears from the union of bags that determine the depth at that bag, especially if a vertex disappears from the bag at $s_i$, then it disappears from every bag in $V(T_i)$.  
	
	\begin{claim}
		\label{cl:remove-vertex-new}
		Let $i\in [n_T]$ and let $t\in V(T_{i-1})$. If $v\in \beta_{i-1}(t)\setminus\beta_{i}(t)$, then $v\notin\beta_i(t^*)$, for all $t^*\in V(T_{i-1})$ 
		such that $t$ is contained in the path from $t^*$ to $s_i$.
	\end{claim}
	
	\begin{claimproof}
		We have $E_{G^\circ}(v)\cap F^i\neq \emptyset$.
		
		Let $t=s_i$.
		As $v\notin \beta_i(s_i)$ we get that $v\notin \delta(\gamma_i(s_i,p_{s_i}))$ and thus 
		$E_{G^\circ}(v)\cap \gamma_i(s_i,p_{s_i}) = E_{G^\circ}(v)\cap \gamma_{i-1}(s_i,p_{s_i}) \cap \complementOf{F^i} = \emptyset$.
		Now let $t^*\in V(T_{i-1})$ and $t'$ be the next node on the path from $t^*$ to $s_i$.
		Then by \cref{lem:exact} we get that $\gamma_i(t^*,t')\supseteq\gamma_i(p_{s_i},s_i)\supseteq E_{G^\circ}(v)$ 
		and thus $v\notin\beta_i(t^*)$.
		
		Otherwise let $t\neq s_i$
		Let $t'$ be the next node on the path from $t$ to $s_i$.
		As $v\notin\delta(\gamma_i(t,t'))$ it follows that $E_{G^\circ}(v)\subseteq \gamma_i(t,t')=\gamma_{i-1}(t,t')\cup F^i$, that $v\in\delta(\gamma_{i-1}(t,t'))$, and that $E_{G^\circ}(v) \cap \gamma_{i-1}(t',t)\subseteq E_{G^\circ}(v)\cap F^i$.
		Assume there is some $t^*\in V(T_{i-1})$ such that $v\in\beta_i(t^*)$.
		We observe that due to \cref{lem:exact} and because all edges incident to $v$ are contained in $\gamma_i(t,t')$, we get that $t$ is not contained in the path from $t^*$ to $s_i$.
	\end{claimproof}
	
	We are now ready to prove the lemma.
	Towards this, let $i\geq 1$ and assume the statement holds for $i-1$.
	We consider all vertices that appear at a bag at any node in $T_i$ due to the changes in step $i$.
	Observe that if $\beta_i(s_i)=\beta_{i-1}(s_i)$, then there are no changes to the bags at other nodes than the $t^i_j$ by minimality of $|F^i|$, and if $\beta_i(s_i)\neq\beta_{i-1}(s_i)$,  we have $|\beta_i(s_i)|<|\beta_{i-1}(s_i)|$ again by the minimality of the choice.
	
	Let $t\in V(T_{i-1})$.
	Let $U \coloneqq \left(\bigcup_{s\in P_t} \beta_{i}(s)\right) \setminus \left(\bigcup_{s\in P_t} \beta_{i-1}(s)\right)$.
	Let $t^*$ be the greatest common ancestor of $t$ and $s_i$.
	As $t^*$ is on every path from some node in $P_t$ to $s_i$, from \cref{cl:new-vertex-new} we know that $u\in\beta_i(t^*)\setminus \beta_{i-1}(t^*)$, for all $u\in U$.
	Let $W = \beta_{i-1}(t^*)\setminus \beta_{i}(t^*)$
	As by \cref{lem:width} $|\beta_i(t^*)|\leq |\beta_{i-1}(t^*)|$, we know that $|U| \leq |W|$.
	By \cref{cl:remove-vertex-new} we get that $W\subseteq \left(\bigcup_{s\in P_t} \beta_{i-1}(s)\right) \setminus \left(\bigcup_{s\in P_t} \beta_{i}(s)\right)$.
	By this \emph{vertex exchange} we conclude that $\left|\bigcup_{s\in P_t} \beta_{i}(s)\right| \leq \left|\bigcup_{s\in P_t} \beta_{i-1}(s)\right|$.
	
	Otherwise $t=t^i_j$ for some $j\in [a_i]$.
	By construction we get $\bigcup_{s\in P_{t}} \beta_{i}(s) = \bigcup_{s\in P_{s_i}} \beta_{i}(s) \cup \beta_i(t) \setminus \beta_i(s_i)$.
	We have shown above that $|\bigcup_{s\in P_{s_i}} \beta_{i}(s)| \leq \sum_{s\in P_{s_i}\setminus\{r\}} |\beta(s) \setminus \beta(p_s)|$ 
	and by \cref{cl:depth-leafs-new} we have $\beta_i(t) \setminus \beta_i(s_i) \subseteq \beta(t)\setminus \beta(s_i)$.
	Thus we can bound the union $|\bigcup_{s\in P_{t}} \beta_i(s)| \leq \sum_{s\in P_{t}\setminus\{r\}} |\beta(s) \setminus \beta(p_s)|$.	
\end{proof}

\begin{proof}[Proof of \cref{thm:exactness}]
	Combining \cref{lem:exact,lem:width,lem:depth} we get that there exists an exact \preTreeDec\ of $G^\circ$ of width $\leq k$ and depth $\leq q$, if the cop player wins $\CRkq[G^\circ]$.
	The theorem then follows from \cref{lem:tw-ptw}.
\end{proof}

Summarising all results we get the following equivalences.

\begin{theorem}
	\label{thm:equivalence}
	Let $k,q\geq 1$ and $G$ be a graph. The following are equivalent:
	\begin{bracketenumerate}
		\item $G$ admits a tree decomposition of width at most $k-1$ and depth at most $q$.\label{ax:tdG}
		\item $G^\circ$ admits a tree decomposition of width at most $k-1$ and depth at most $q$.\label{ax:tdGc}
		\item $G^\circ$ admits an exact \preTreeDec\ of width at most $k-1$ and depth at most $q$.\label{ax:ptdGc}
		\item The cop player wins $\monCRkq[G^\circ]$.\label{ax:mcpGc}
		\item The cop player wins $\CRkq[G^\circ]$.\label{ax:cpGc}
		\item The cop player wins $\monCRkq$.\label{ax:mcpG}
		\item The cop player wins $\CRkq$.\label{ax:cpG}
	\end{bracketenumerate}
\end{theorem}

\begin{proof}
	Every tree decomposition of $G$ is also a tree decomposition of $G^\circ$ and vice versa, thus (\ref{ax:tdG}) and (\ref{ax:tdGc}) are equivalent.
	\Cref{lem:tw-ptw} shows the equivalence of (\ref{ax:tdG}) and (\ref{ax:ptdGc}).
	\Cref{thm:exactness} shows that (\ref{ax:cpGc}) implies (\ref{ax:tdG}).
	Let $G$ be a graph without edges.
	The cop player wins on $G$ before the first round and on $G^\circ$ after one placement with one cop.
	On a connected component with at least one edge the robber escape spaces are essentially the same.
	Thus for $q,k\geq 1$, by construction, (\ref{ax:mcpGc}) and (\ref{ax:mcpG}) as well as (\ref{ax:cpGc}) and (\ref{ax:cpG}) are equivalent.
	\Cref{lem:moncr-vs-tkq} shows that (\ref{ax:tdG}) is equivalent to (\ref{ax:mcpG}) and thus also (\ref{ax:tdGc}) to (\ref{ax:mcpGc}).
	Furthermore (\ref{ax:mcpG}) implies (\ref{ax:cpG}) and (\ref{ax:mcpGc}) implies (\ref{ax:cpGc}) by construction.
	Thus the theorem follows.
\end{proof}

\section{Excursion on counting homomorphisms}
\label{sec:hom}

In this section we give an overview over the field of counting homomorphisms and the equivalence relations on graphs, that can be derived from these counts.
We focus ourselves to the results and open questions regarding the homomorphism counts from graphs in the class $\mathcal{T}^k_q$, for fixed $k,q\geq 0$.

We start by recalling the important definitions.
Let $G,F$ be two graphs.
A \emph{homomorphism} from $F$ into $G$ is a function $\varphi \colon V(F) \rightarrow V(G)$, such that for every $uv \in E(F)$, it holds that $\varphi(u) \varphi(v) \in E(G)$.
By $\operatorname{hom}(F,G)$ we denote the number of homomorphisms from $F$ into $G$.
Let $\mathcal{F}$ be a graph class.
We say two graphs $G$ and $H$ are \emph{homomorphism indistinguishable over $\mathcal{F}$} if, for every $F \in \mathcal{F}$ it holds that $\operatorname{hom}(F,G) = \operatorname{hom}(F,H)$, we write $G \equiv_\mathcal{F} H$.
A graph class $\mathcal{F}$ is \emph{homomorphism distinguishing closed} if, for every $F\notin \mathcal{F}$, there exist two graphs $G,H$ such that $G \equiv_\mathcal{F} H$, but $\operatorname{hom}(F,G) \neq \operatorname{hom}(F,H)$.

In \cite{FluckSS24} the authors have reduced the question whether the class $\mathcal{T}^k_q$ is homomorphism distinguishing closed down to the question if monotonicity is a restriction for the cop player.
The definition of the cops-and-robber game the authors use is slightly different.
In their game the robber hides in vertices and is caught if a cops occupies the same vertex.
With the same arguments as in the proof of \cref{thm:equivalence} we observe that these games are equivalent.
The following lemma is thus implied in \cite{FluckSS24}.

\begin{lemma}[\cite{FluckSS24}]
	\label{lem:cop-hdc}
	Let $k,q\geq 1$.
	The graph class $\mathcal{C}\coloneqq \{G\mid \text{cop player wins }\CRkq\}$ is homomorphism distinguishing closed.
\end{lemma}

In this paper we show that the cop player wins \CRkq if and only if $G\in \mathcal{T}^k_q$, thus the we get the following.

\begin{theorem}
	Let $k,q\geq 0$ be integers.
	The class $\mathcal{T}^k_q$ is homomorphism distinguishing closed.
\end{theorem}

\begin{proof}
	Assume $k,q\geq 1$.
	Then theorem follows directly from \cref{lem:cop-hdc,thm:equivalence}.
	Thus assume $k=0$ or $q=0$.
	The only graph that has a tree decomposition of depth $0$ or width $-1$ is the empty graph.
	The only homomorphism from the empty graph into any graph is the empty function, which is a homomorphism independent of the right-hand-side graph.
	Thus there are no graphs that can be distinguished by the number of homomorphisms from the empty graph.
	Thus the theorem holds.
\end{proof}

\section{Conclusion}
We gave a new characterisation of bounded depth treewidth by the cops and robber game with both a bound on the number of cops and on the number of placements,where the cops are allowed to make non-monotone moves.
As a corollary we gave a positive answer to an open question on homomorphism counts.
The core of our contribution is a proof of monotonicity of this game.
For this proof we substantially reorganise a winning strategy.
First we transform it into a pre-decomposition.
Then we apply a breadth-first `cleaning up' procedure along the pre-decomposition (which may temporarily lose the property of representing a strategy), in order to achieve monotonicity while controlling the number of cop placements simultaneously across all branches of the decomposition via a vertex exchange argument(cf.~the proof of~\cref{lem:depth}).
As an interesting observation we obtain that cop moves into the back country, i.\,e.~to positions that are not part of the boundary, can be ignored and the depth of the exact \preTreeDec\ is the number of cops placed into the robber escape space:
We observe that in the proof of \cref{cl:depth-leafs-new} where we compute how much larger the depth at some node $t^i_j$ at step $i$ is than at the \considered\ node $s_i$ the depth increases only if the node $t^i_j$ is branching by \cref{obs:self-loops} as $t^i_j$ has a child where the cone contains only a self-loop and hence this is a move into the robber space.
\begin{corollary}
	$\dep(T,r,\beta_{n_T},\gamma_{n_T})\leq \max_{\ell\in L(T)} |\{t\in P_\ell\mid t \text{ is branching} \}|$.
\end{corollary}

In the future, it would be interesting to know if it is possible to give a proof that entirely argues with game strategies (not requiring pre-decopmositions), and we leave this open.
We also leave open whether a dual object similar to brambles can be defined for bounded depth treewidth. 
Finally, given a winning strategy for $k$ cops with $q$ placements,  
it would be interesting to know if it is possible to bound the number of cops necessary for winning with only $q-1$ placements in terms of $k$ and $q$, given that the cop player still can win.

\bibliography{mon}

\begin{thebibliography}{10}

\bibitem{Abramsky_pebbling_2017}
Samson Abramsky, Anuj Dawar, and Pengming Wang.
\newblock The pebbling comonad in finite model theory.
\newblock In {\em 32nd Annual {ACM/IEEE} Symposium on Logic in Computer
  Science, {LICS} 2017, Reykjavik, Iceland, June 20-23, 2017}, pages 1--12.
  {IEEE} Computer Society, 2017.
\newblock \href {https://doi.org/10.1109/LICS.2017.8005129}
  {\path{doi:10.1109/LICS.2017.8005129}}.

\bibitem{Adler04}
Isolde Adler.
\newblock Marshals, monotone marshals, and hypertree-width.
\newblock {\em J. Graph Theory}, 47(4):275--296, 2004.
\newblock URL: \url{https://doi.org/10.1002/jgt.20025}, \href
  {https://doi.org/10.1002/JGT.20025} {\path{doi:10.1002/JGT.20025}}.

\bibitem{abs-0906-3857}
Isolde Adler.
\newblock Games for width parameters and monotonicity.
\newblock {\em CoRR}, abs/0906.3857, 2009.
\newblock URL: \url{http://arxiv.org/abs/0906.3857}, \href
  {http://arxiv.org/abs/0906.3857} {\path{arXiv:0906.3857}}.

\bibitem{AdlerGG07}
Isolde Adler, Georg Gottlob, and Martin Grohe.
\newblock Hypertree width and related hypergraph invariants.
\newblock {\em Eur. J. Comb.}, 28(8):2167--2181, 2007.
\newblock URL: \url{https://doi.org/10.1016/j.ejc.2007.04.013}, \href
  {https://doi.org/10.1016/J.EJC.2007.04.013}
  {\path{doi:10.1016/J.EJC.2007.04.013}}.

\bibitem{AignerF84}
Martin Aigner and M.~Fromme.
\newblock A game of cops and robbers.
\newblock {\em Discret. Appl. Math.}, 8(1):1--12, 1984.
\newblock \href {https://doi.org/10.1016/0166-218X(84)90073-8}
  {\path{doi:10.1016/0166-218X(84)90073-8}}.

\bibitem{AminiMNT09}
Omid Amini, Fr{\'{e}}d{\'{e}}ric Mazoit, Nicolas Nisse, and St{\'{e}}phan
  Thomass{\'{e}}.
\newblock Submodular partition functions.
\newblock {\em Discret. Math.}, 309(20):6000--6008, 2009.
\newblock URL: \url{https://doi.org/10.1016/j.disc.2009.04.033}, \href
  {https://doi.org/10.1016/J.DISC.2009.04.033}
  {\path{doi:10.1016/J.DISC.2009.04.033}}.

\bibitem{Bienstock89}
Daniel Bienstock.
\newblock Graph searching, path-width, tree-width and related problems {(A}
  survey).
\newblock In Fred Roberts, Frank Hwang, and Clyde~L. Monma, editors, {\em
  Reliability Of Computer And Communication Networks, Proceedings of a {DIMACS}
  Workshop, New Brunswick, New Jersey, USA, December 2-4, 1989}, volume~5 of
  {\em {DIMACS} Series in Discrete Mathematics and Theoretical Computer
  Science}, pages 33--50. {DIMACS/AMS}, 1989.
\newblock URL: \url{https://doi.org/10.1090/dimacs/005/02}, \href
  {https://doi.org/10.1090/DIMACS/005/02} {\path{doi:10.1090/DIMACS/005/02}}.

\bibitem{BienstockS91}
Daniel Bienstock and Paul~D. Seymour.
\newblock Monotonicity in graph searching.
\newblock {\em J. Algorithms}, 12(2):239--245, 1991.
\newblock \href {https://doi.org/10.1016/0196-6774(91)90003-H}
  {\path{doi:10.1016/0196-6774(91)90003-H}}.

\bibitem{BodlaenderT04}
Hans~L. Bodlaender and Dimitrios~M. Thilikos.
\newblock Computing small search numbers in linear time.
\newblock In Rodney~G. Downey, Michael~R. Fellows, and Frank K. H.~A. Dehne,
  editors, {\em Parameterized and Exact Computation, First International
  Workshop, {IWPEC} 2004, Bergen, Norway, September 14-17, 2004, Proceedings},
  volume 3162 of {\em Lecture Notes in Computer Science}, pages 37--48.
  Springer, 2004.
\newblock \href {https://doi.org/10.1007/978-3-540-28639-4\_4}
  {\path{doi:10.1007/978-3-540-28639-4\_4}}.

\bibitem{dawar_lovasz-type_2021}
Anuj Dawar, Tomá\v{s} Jakl, and Luca Reggio.
\newblock Lovász-{Type} {Theorems} and {Game} {Comonads}.
\newblock In {\em 2021 36th {Annual} {ACM}/{IEEE} {Symposium} on {Logic} in
  {Computer} {Science} ({LICS})}, pages 1--13, June 2021.
\newblock \href {https://doi.org/10.1109/LICS52264.2021.9470609}
  {\path{doi:10.1109/LICS52264.2021.9470609}}.

\bibitem{dell_lovasz_2018}
Holger Dell, Martin Grohe, and Gaurav Rattan.
\newblock Lovász {Meets} {Weisfeiler} and {Leman}.
\newblock {\em 45th International Colloquium on Automata, Languages, and
  Programming (ICALP 2018)}, pages 40:1--40:14, 2018.
\newblock \href {https://doi.org/10.4230/LIPICS.ICALP.2018.40}
  {\path{doi:10.4230/LIPICS.ICALP.2018.40}}.

\bibitem{dvorak_recognizing_2010}
Zden{\v e}k Dvo{\v r}{\'a}k.
\newblock On recognizing graphs by numbers of homomorphisms.
\newblock {\em Journal of Graph Theory}, 64(4):330--342, August 2010.
\newblock \href {https://doi.org/10.1002/jgt.20461}
  {\path{doi:10.1002/jgt.20461}}.

\bibitem{FluckSS24}
Eva Fluck, Tim Seppelt, and Gian~Luca Spitzer.
\newblock {Going Deep and Going Wide: Counting Logic and Homomorphism
  Indistinguishability over Graphs of Bounded Treedepth and Treewidth}.
\newblock In Aniello Murano and Alexandra Silva, editors, {\em 32nd EACSL
  Annual Conference on Computer Science Logic (CSL 2024)}, volume 288 of {\em
  Leibniz International Proceedings in Informatics (LIPIcs)}, pages
  27:1--27:17, Dagstuhl, Germany, 2024. Schloss Dagstuhl -- Leibniz-Zentrum
  f{\"u}r Informatik.
\newblock URL:
  \url{https://drops.dagstuhl.de/entities/document/10.4230/LIPIcs.CSL.2024.27},
  \href {https://doi.org/10.4230/LIPIcs.CSL.2024.27}
  {\path{doi:10.4230/LIPIcs.CSL.2024.27}}.

\bibitem{FominFN09}
Fedor~V. Fomin, Pierre Fraigniaud, and Nicolas Nisse.
\newblock Nondeterministic graph searching: From pathwidth to treewidth.
\newblock {\em Algorithmica}, 53(3):358--373, 2009.
\newblock URL: \url{https://doi.org/10.1007/s00453-007-9041-6}, \href
  {https://doi.org/10.1007/S00453-007-9041-6}
  {\path{doi:10.1007/S00453-007-9041-6}}.

\bibitem{FominGK08}
Fedor~V. Fomin, Petr~A. Golovach, and Jan Kratochv{\'{\i}}l.
\newblock On tractability of cops and robbers game.
\newblock In Giorgio Ausiello, Juhani Karhum{\"{a}}ki, Giancarlo Mauri, and
  C.{-}H.~Luke Ong, editors, {\em Fifth {IFIP} International Conference On
  Theoretical Computer Science - {TCS} 2008, {IFIP} 20th World Computer
  Congress, {TC} 1, Foundations of Computer Science, September 7-10, 2008,
  Milano, Italy}, volume 273 of {\em {IFIP}}, pages 171--185. Springer, 2008.
\newblock \href {https://doi.org/10.1007/978-0-387-09680-3\_12}
  {\path{doi:10.1007/978-0-387-09680-3\_12}}.

\bibitem{FominT08}
Fedor~V. Fomin and Dimitrios~M. Thilikos.
\newblock An annotated bibliography on guaranteed graph searching.
\newblock {\em Theor. Comput. Sci.}, 399(3):236--245, 2008.
\newblock URL: \url{https://doi.org/10.1016/j.tcs.2008.02.040}, \href
  {https://doi.org/10.1016/J.TCS.2008.02.040}
  {\path{doi:10.1016/J.TCS.2008.02.040}}.

\bibitem{FranklinGY00}
Matthew~K. Franklin, Zvi Galil, and Moti Yung.
\newblock Eavesdropping games: a graph-theoretic approach to privacy in
  distributed systems.
\newblock {\em J. {ACM}}, 47(2):225--243, 2000.
\newblock \href {https://doi.org/10.1145/333979.333980}
  {\path{doi:10.1145/333979.333980}}.

\bibitem{GiannopoulouHT12}
Archontia~C. Giannopoulou, Paul Hunter, and Dimitrios~M. Thilikos.
\newblock Lifo-search: {A} min-max theorem and a searching game for cycle-rank
  and tree-depth.
\newblock {\em Discret. Appl. Math.}, 160(15):2089--2097, 2012.
\newblock URL: \url{https://doi.org/10.1016/j.dam.2012.03.015}, \href
  {https://doi.org/10.1016/J.DAM.2012.03.015}
  {\path{doi:10.1016/J.DAM.2012.03.015}}.

\bibitem{GiannopoulouT11}
Archontia~C. Giannopoulou and Dimitrios~M. Thilikos.
\newblock A min-max theorem for lifo-search.
\newblock {\em Electron. Notes Discret. Math.}, 38:395--400, 2011.
\newblock URL: \url{https://doi.org/10.1016/j.endm.2011.09.064}, \href
  {https://doi.org/10.1016/J.ENDM.2011.09.064}
  {\path{doi:10.1016/J.ENDM.2011.09.064}}.

\bibitem{grohe_counting_2020}
Martin Grohe.
\newblock Counting {Bounded} {Tree} {Depth} {Homomorphisms}.
\newblock In {\em Proceedings of the 35th {Annual} {ACM}/{IEEE} {Symposium} on
  {Logic} in {Computer} {Science}}, {LICS} '20, pages 507--520, New York, NY,
  USA, 2020. Association for Computing Machinery.
\newblock event-place: Saarbr{\"u}cken, Germany.
\newblock \href {https://doi.org/10.1145/3373718.3394739}
  {\path{doi:10.1145/3373718.3394739}}.

\bibitem{GroheM14}
Martin Grohe and D{\'{a}}niel Marx.
\newblock Constraint solving via fractional edge covers.
\newblock {\em {ACM} Trans. Algorithms}, 11(1):4:1--4:20, 2014.
\newblock \href {https://doi.org/10.1145/2636918} {\path{doi:10.1145/2636918}}.

\bibitem{grohe_homomorphism_2022}
Martin Grohe, Gaurav Rattan, and Tim Seppelt.
\newblock {Homomorphism Tensors and Linear Equations}.
\newblock In Miko{\l}aj Boja\'{n}czyk, Emanuela Merelli, and David~P. Woodruff,
  editors, {\em 49th International Colloquium on Automata, Languages, and
  Programming (ICALP 2022)}, volume 229 of {\em Leibniz International
  Proceedings in Informatics (LIPIcs)}, pages 70:1--70:20, Dagstuhl, Germany,
  2022. Schloss Dagstuhl -- Leibniz-Zentrum f{\"u}r Informatik.
\newblock \href {https://doi.org/10.4230/LIPIcs.ICALP.2022.70}
  {\path{doi:10.4230/LIPIcs.ICALP.2022.70}}.

\bibitem{HollingerKS10}
Geoffrey~A. Hollinger, Athanasios Kehagias, and Sanjiv Singh.
\newblock {GSST:} anytime guaranteed search.
\newblock {\em Auton. Robots}, 29(1):99--118, 2010.
\newblock URL: \url{https://doi.org/10.1007/s10514-010-9189-9}, \href
  {https://doi.org/10.1007/S10514-010-9189-9}
  {\path{doi:10.1007/S10514-010-9189-9}}.

\bibitem{HunterK08}
Paul Hunter and Stephan Kreutzer.
\newblock Digraph measures: Kelly decompositions, games, and orderings.
\newblock {\em Theor. Comput. Sci.}, 399(3):206--219, 2008.
\newblock URL: \url{https://doi.org/10.1016/j.tcs.2008.02.038}, \href
  {https://doi.org/10.1016/J.TCS.2008.02.038}
  {\path{doi:10.1016/J.TCS.2008.02.038}}.

\bibitem{JohnsonRST01}
Thor Johnson, Neil Robertson, Paul~D. Seymour, and Robin Thomas.
\newblock Directed tree-width.
\newblock {\em J. Comb. Theory, Ser. {B}}, 82(1):138--154, 2001.
\newblock URL: \url{https://doi.org/10.1006/jctb.2000.2031}, \href
  {https://doi.org/10.1006/JCTB.2000.2031} {\path{doi:10.1006/JCTB.2000.2031}}.

\bibitem{LaPaugh93}
Andrea~S. LaPaugh.
\newblock Recontamination does not help to search a graph.
\newblock {\em J. {ACM}}, 40(2):224--245, 1993.
\newblock \href {https://doi.org/10.1145/151261.151263}
  {\path{doi:10.1145/151261.151263}}.

\bibitem{lovasz_operations_1967}
L{\'a}szl{\'o} Lov{\'a}sz.
\newblock Operations with structures.
\newblock {\em Acta Mathematica Academiae Scientiarum Hungarica},
  18(3):321--328, September 1967.
\newblock \href {https://doi.org/10.1007/BF02280291}
  {\path{doi:10.1007/BF02280291}}.

\bibitem{MakedonS89}
Fillia Makedon and Ivan~Hal Sudborough.
\newblock On minimizing width in linear layouts.
\newblock {\em Discret. Appl. Math.}, 23(3):243--265, 1989.
\newblock \href {https://doi.org/10.1016/0166-218X(89)90016-4}
  {\path{doi:10.1016/0166-218X(89)90016-4}}.

\bibitem{mancinska_quantum_2020}
Laura Man{\v c}inska and David~E. Roberson.
\newblock Quantum isomorphism is equivalent to equality of homomorphism counts
  from planar graphs.
\newblock In {\em 2020 {IEEE} 61st {Annual} {Symposium} on {Foundations} of
  {Computer} {Science} ({FOCS})}, pages 661--672, 2020.
\newblock \href {https://doi.org/10.1109/FOCS46700.2020.00067}
  {\path{doi:10.1109/FOCS46700.2020.00067}}.

\bibitem{MazoitN08}
Fr{\'{e}}d{\'{e}}ric Mazoit and Nicolas Nisse.
\newblock Monotonicity of non-deterministic graph searching.
\newblock {\em Theor. Comput. Sci.}, 399(3):169--178, 2008.
\newblock URL: \url{https://doi.org/10.1016/j.tcs.2008.02.036}, \href
  {https://doi.org/10.1016/J.TCS.2008.02.036}
  {\path{doi:10.1016/J.TCS.2008.02.036}}.

\bibitem{NesetrilM06}
Jaroslav Nesetril and Patrice~Ossona de~Mendez.
\newblock Tree-depth, subgraph coloring and homomorphism bounds.
\newblock {\em Eur. J. Comb.}, 27(6):1022--1041, 2006.
\newblock URL: \url{https://doi.org/10.1016/j.ejc.2005.01.010}, \href
  {https://doi.org/10.1016/J.EJC.2005.01.010}
  {\path{doi:10.1016/J.EJC.2005.01.010}}.

\bibitem{neuen_homomorphism-distinguishing_2023}
Daniel Neuen.
\newblock Homomorphism-{Distinguishing} {Closedness} for {Graphs} of {Bounded}
  {Tree}-{Width}, April 2023.
\newblock \href {https://doi.org/10.48550/arXiv.2304.07011}
  {\path{doi:10.48550/arXiv.2304.07011}}.

\bibitem{Obdrzalek06}
Jan Obdrz{\'{a}}lek.
\newblock Dag-width: connectivity measure for directed graphs.
\newblock In {\em Proceedings of the Seventeenth Annual {ACM-SIAM} Symposium on
  Discrete Algorithms, {SODA} 2006, Miami, Florida, USA, January 22-26, 2006},
  pages 814--821. {ACM} Press, 2006.
\newblock URL: \url{http://dl.acm.org/citation.cfm?id=1109557.1109647}.

\bibitem{Parsons78}
T.~D. Parsons.
\newblock Pursuit-evasion in a graph.
\newblock In Yousef Alavi and Don~R. Lick, editors, {\em Theory and
  Applications of Graphs}, pages 426--441, Berlin, Heidelberg, 1978. Springer
  Berlin Heidelberg.

\bibitem{parsons78search}
Torrence~D Parsons.
\newblock The search number of a connected graph.
\newblock In {\em Proc. 9th South-Eastern Conf. on Combinatorics, Graph Theory,
  and Computing}, pages 549--554, 1978.

\bibitem{petrov82}
Nikolai~N. Petrov.
\newblock A problem of pursuit in the absence of information on the pursued.
\newblock {\em Differentsial'nye Uravneniya}, 18(1):345–--1352, 1982.

\bibitem{Roberson_oddomorphisms_2022}
David~E. Roberson.
\newblock Oddomorphisms and homomorphism indistinguishability over graphs of
  bounded degree, June 2022.
\newblock \href {https://doi.org/10.48550/arXiv.2206.10321}
  {\path{doi:10.48550/arXiv.2206.10321}}.

\bibitem{roberson_lasserre_2023}
David~E. Roberson and Tim Seppelt.
\newblock {Lasserre Hierarchy for Graph Isomorphism and Homomorphism
  Indistinguishability}.
\newblock In Kousha Etessami, Uriel Feige, and Gabriele Puppis, editors, {\em
  50th International Colloquium on Automata, Languages, and Programming (ICALP
  2023)}, volume 261 of {\em Leibniz International Proceedings in Informatics
  (LIPIcs)}, pages 101:1--101:18, Dagstuhl, Germany, 2023. Schloss Dagstuhl --
  Leibniz-Zentrum f{\"u}r Informatik.
\newblock \href {https://doi.org/10.4230/LIPIcs.ICALP.2023.101}
  {\path{doi:10.4230/LIPIcs.ICALP.2023.101}}.

\bibitem{ScheidtS23}
Benjamin Scheidt and Nicole Schweikardt.
\newblock Counting homomorphisms from hypergraphs of bounded generalised
  hypertree width: {A} logical characterisation.
\newblock In J{\'{e}}r{\^{o}}me Leroux, Sylvain Lombardy, and David Peleg,
  editors, {\em 48th International Symposium on Mathematical Foundations of
  Computer Science, {MFCS} 2023, August 28 to September 1, 2023, Bordeaux,
  France}, volume 272 of {\em LIPIcs}, pages 79:1--79:15. Schloss Dagstuhl -
  Leibniz-Zentrum f{\"{u}}r Informatik, 2023.
\newblock URL: \url{https://doi.org/10.4230/LIPIcs.MFCS.2023.79}, \href
  {https://doi.org/10.4230/LIPICS.MFCS.2023.79}
  {\path{doi:10.4230/LIPICS.MFCS.2023.79}}.

\bibitem{seppelt_logical_2023}
Tim Seppelt.
\newblock {Logical Equivalences, Homomorphism Indistinguishability, and
  Forbidden Minors}.
\newblock In J\'{e}r\^{o}me Leroux, Sylvain Lombardy, and David Peleg, editors,
  {\em 48th International Symposium on Mathematical Foundations of Computer
  Science (MFCS 2023)}, volume 272 of {\em Leibniz International Proceedings in
  Informatics (LIPIcs)}, pages 82:1--82:15, Dagstuhl, Germany, 2023. Schloss
  Dagstuhl -- Leibniz-Zentrum f{\"u}r Informatik.
\newblock \href {https://doi.org/10.4230/LIPIcs.MFCS.2023.82}
  {\path{doi:10.4230/LIPIcs.MFCS.2023.82}}.

\bibitem{SeymourT93}
Paul~D. Seymour and Robin Thomas.
\newblock Graph searching and a min-max theorem for tree-width.
\newblock {\em J. Comb. Theory, Ser. {B}}, 58(1):22--33, 1993.
\newblock URL: \url{https://doi.org/10.1006/jctb.1993.1027}, \href
  {https://doi.org/10.1006/JCTB.1993.1027} {\path{doi:10.1006/JCTB.1993.1027}}.

\end{thebibliography}
\bibliographystyle{plainurl}

\appendix

\end{document}